\documentclass[journal,onecolumn,draftcls]{IEEEtran}

\usepackage{amsmath,amsfonts}
\usepackage{algorithmic}
\usepackage{array}
\usepackage[caption=false,font=footnotesize,labelfont=rm,textfont=rm]{subfig}

\usepackage{textcomp}
\usepackage{stfloats}
\usepackage{url}
\usepackage{verbatim}
\usepackage{graphicx}
\usepackage{balance}

\usepackage{amsmath}
\usepackage{makecell}
\usepackage{caption}
\usepackage{subcaption}

\usepackage{longtable}
\usepackage{lipsum}
\usepackage{tabulary}
\usepackage{graphicx}
\usepackage{subfiles}

\usepackage[open,openlevel=3,atend,numbered]{bookmark}

\usepackage[table,dvipsnames]{xcolor}
\usepackage{listings}

\usepackage{enumitem}
\usepackage{array}
\usepackage{multirow}
\usepackage{makecell}
\usepackage[export]{adjustbox}
\usepackage{amsthm,amsmath,amssymb}
\usepackage{mathrsfs}

\usepackage{algorithm,algorithmic}

\usepackage{orcidlink}
\newcounter{magicrownumbers}

\DeclareMathOperator*{\argmin}{argmin}
\DeclareMathOperator*{\argmax}{argmax}

\usepackage{accents}
\newcommand{\set}[1]{\mathcal{#1}}

\newcommand{\rv}[1]{\mathsf{#1}}
\newcommand{\map}[1]{\mathsf{#1}}

\newcommand\ie{\textit{i.e.}}
\newcommand\eg{\textit{e.g.}}

\usepackage[backend=bibtex,style=ieee,sorting=none,url = false, doi = false,isbn = false]{biblatex}
\AtEveryBibitem{\clearfield{language}\clearlist{language}\clearfield{langid}\clearlist{langid}} 
\newtheoremstyle{thm_custom}
{3pt} % 上方空间
{3pt} % 下方空间
{\itshape} % 正文字体
{} % 缩进
{\itshape} % 定理头部字体
{:} % 定理头部后标点
{.5em} % 标题后额外空间
{\thmname{#1}\thmnumber{\textit{ #2}}\thmnote{ (#3)}}   % 定理头部格式

\newtheoremstyle{def_custom}
{3pt} % 上方空间
{3pt} % 下方空间
{} % 正文字体
{} % 缩进
{\itshape} % 定理头部字体
{:} % 定理头部后标点
{.5em} % 标题后额外空间
{\thmname{#1}\thmnumber{\textit{ #2}}\thmnote{ (#3)}}   % 定理头部格式

\theoremstyle{thm_custom}
\newtheorem{theorem}{Theorem}
\newtheorem{lemma}{Lemma}
\newtheorem{corollary}{Corollary}
\newtheorem{proposition}{Proposition}

\theoremstyle{def_custom}
\newtheorem{definition}{Definition}
    			
\theoremstyle{remark}
\newtheorem{remark}{Remark}

\usepackage{hyperref}
\hypersetup{colorlinks = true, linkcolor = blue, citecolor = blue}
\usepackage[backend=bibtex,style=ieee,sorting=none,url = false, doi = false,isbn = false]{biblatex}
\AtEveryBibitem{\clearfield{language}\clearlist{language}\clearfield{langid}\clearlist{langid}} 
\begin{document}
\title{Integrated Sensing and Communication: Rate-Distortion Fundamental Limits of State Estimator \\
\thanks{}
}

%% Author& Thanks & Markboth
\author{
Lugaoze Feng$^{\orcidlink{0009-0000-4014-4154}}$,
Guocheng Lv$^{\orcidlink{0000-0002-7136-3402}}$,
Xunan Li$^{\orcidlink{0000-0002-5740-161X}}$, and
Ye jin
\thanks{}
\thanks{Lugaoze Feng, Guocheng Lv and Ye Jin are with the State Key Laboratory of Photonics and Communications, Peking University, Beijing 100871, China (e-mail: lgzf@stu.pku.edu.cn; lv.guocheng@pku.edu.cn; jinye@pku.edu.cn).}
\thanks{Xunan Li is with the National Computer Network Emergency Response Technical Team/Coordination Center of China, Beijing 100029, China (e-mail: lixunan@cert.org.cn).}}

\markboth
{ }
{ }

\maketitle

\begin{abstract}
	The state-dependent memoryless channel (SDMC) is employed to model the integrated sensing and communication (ISAC) system, where the transmitter conveys messages to the receiver while simultaneously estimating the state parameter of interest via the received echo signals. However, the performance of sensing has often been neglected in existing works. To address this gap, we establish the rate-distortion function for sensing performance in the SDMC model, which is defined based on standard information-theoretic principles to ensure clear operational meaning. In addition, we propose a modified Blahut-Arimoto type algorithm for solving the rate-distortion function and provide convergence proofs for the algorithm. We further define the capacity-rate-distortion tradeoff region, which unifies information-theoretic results for communication and sensing within a single optimization framework. Finally, we numerically evaluate the capacity-rate-distortion region and demonstrate the benefit of coding in terms of estimation rate for certain channels.
\end{abstract}

\begin{IEEEkeywords}
Integrated sensing and communication, rate-distortion theory.
\end{IEEEkeywords}
	
\section{Introduction}
\IEEEPARstart{F}{uture} communication networks are poised to integrate sensing and communication functionalities into unified systems \cite{liu2022survey}, \cite{liu2022integrated}, \cite{wang_2022_joint}. For instance, autonomous vehicles acquire real-time environmental data through integrated networks, which facilitates navigation and helps avoid traffic congestion \cite{zheng2015reliable}. This paradigm, termed \textit{Integrated Sensing and Communication} (ISAC), leverages shared spectrum and hardware so that the same radio signals are used both to transmit data and to sense the environment, thereby enabling simultaneous message delivery and receiver-state estimation. 

To elucidate the tradeoff between sensing and communication performance in ISAC systems, several previous studies have made progress. For example, the ``estimation entropy'' is proposed to quantify the sensing estimation performance and communication rate \cite{bliss_cooperative_2014}. Additionally, the Cramer-Rao bound (CRB) serves as a performance metric for sensing, evaluating the tradeoff between mean square error (MSE) and communication rate, as discussed in \cite{xiong_fundamental_2023}. More recently, an ISAC system model, called the state-dependent memoryless channel (SDMC) model, was introduced in \cite{kobayashi2018joint}. This model features a memoryless channel with independent and identically distributed (i.i.d.) state sequences, assuming the receiver has perfect knowledge of the channel state and strictly causal feedback.  In this framework, the author considers the \textit{capacity-distortion tradeoff} as a performance measure to balance the maximizing communication rate and target distortion. The capacity-distortion tradeoff is extended to the continuous channels in \cite{li_computation_2025}. Additionally, this model is extended to the broadcast channel and the multiple access channel scenario in \cite{ahmadipour_information-theoretic_2024} and \cite{liu_fundamental_2025}, respectively. However, the performance measure of the estimator in these works is ignored. Several recent works have attempted to address this issue, including the radar estimation information rate \cite{chiriyath_2016_inner} and the sensing estimation rate \cite{dong_rethinking_2023}. Nevertheless, the operational meaning of these definitions remains insufficiently characterized. 

In the sensing system, the parameters we are interested in include the radar waveform $\rv{X}$, the states corresponding to the target parameter $\rv{S}$, and the noise observation $\rv{T}$, \eg, see \cite{bliss_cooperative_2014}, \cite{bell_information_1993}, \cite{tang_spectrally_2019}, \cite{yang_mimo_2007}, \cite{li_2008_range_compre}, \cite{zhao_2008_iter}, \cite{dong_rethinking_2023}. This system can be characterized by the Markov chain $(\rv{X},\rv{S}) \rightarrow \rv{T}$. As a result, the radar estimator observes a corrupted state sequence. This situation is common in literature related to radar waveform design, and studies using mutual information for radar waveform design have demonstrated superior performance \cite{bell_information_1993}, \cite{tang_spectrally_2019}, \cite{yang_mimo_2007}. When the waveform $\rv{X}$ is given, maximizing the conditional mutual information $I(\rv{S};\rv{T}|\rv{X})$ between the observation $\rv{T}$ and the target parameter $\rv{S}$ leads to the same water-filling solution as minimizing the minimum MSE (MMSE) \cite{yang_mimo_2007}. 

In the context of ISAC systems, we aim to determine, in the sense of information theory, whether the waveform design that maximizes the conditional mutual information $I(\rv{S};\rv{T}|\rv{X})$ mentioned above is optimal. The rate-distortion theory characterizes the minimal source coding rate compatible with a given target average distortion. Indeed, Berger mentions this idea in \cite[p. 9]{berger2003rate} that: ``Rate-distortion theory provides knowledge about how the frequency of faulty categorization will vary with the number and quality of observations''. 
For state sequences corresponding to the parameter of interest, the estimator is reconstructed by obtaining the noise-corrupted sequence while determining the minimum achievable rate that satisfies the target distortion.  The codeword acts as side information in the estimator and acts as noise affecting the state sequence at the receiver, as illustrated in Figure \ref{sys}. Previous information-theoretical work has focused on the noisy source setting without side information \cite{dobrushin_information_1962},\cite{sakrison_source_1968},\cite{ wolf_transmission_1970},\cite{ayanoglu_optimal_1990},\cite{ephraim_unified_nodate}, \cite{kostina_nonasymptotic_2016-1}. 
These works imply that, in the limit of infinite blocklengths, the problem is equivalent to a conventional (noiseless) lossy source coding problem with the surrogate distortion measure. On the other hand, the source coding with side information at the source decoder in the lossy case is studied by Wyner and Ziv's seminal work \cite{wyner2003rate}. The side information at multiple receivers is further studied in \cite{heegard2003rate}. More recent work studied a model in which the side information is corrupted by noise and is available only at the decoder, differing from the ISAC model \cite{feng_functional_nodate}. In the coding theory literature, high-rate quantization and transform coding with side information available at the decoder for a noisy source is established in \cite{rebollo-monedero_wyner-ziv_2004}. Nevertheless, it does not contain information-theoretic results applicable to the ISAC model.
 
 \begin{figure*}[t]
 	\centering
 	\normalsize
 	\includegraphics[width = 0.7\textwidth]{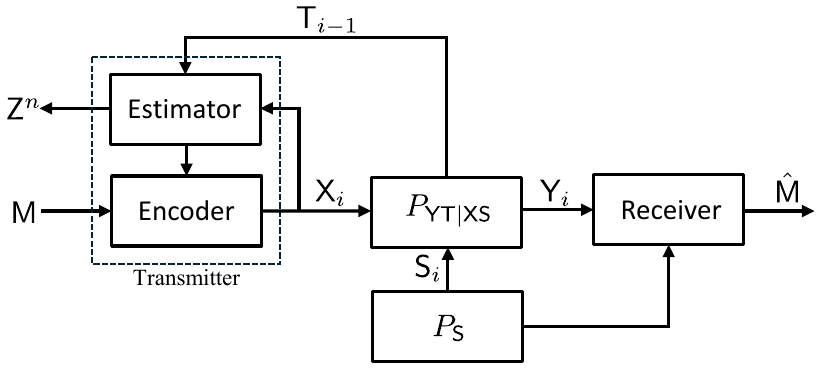}
 	\caption{\ State-dependent memoryless channel model.}
 	\label{sys}
 \end{figure*}

In light of these facts, we initiate an information-theoretic study of rate-distortion theory between the state and its estimation in ISAC systems. The main contributions of this paper are as follows.
\begin{enumerate}
	\item[$\bullet$] We first established the information-theoretic result in the SDMC model, which provides the relationship between the estimation rate of the estimator and the average distortion of the target. The proof of this conclusion is inspired by \cite{ayanoglu_optimal_1990} and is based on the optimal estimator function.
	
	\item[$\bullet$] A modified Blahut-Arimoto algorithm \cite{blahut_computation_1972,csiszar_computation_1974} is proposed to compute the rate-distortion function numerically. We proved the optimality of this alternating optimization and demonstrated that it does indeed converge to the rate-distortion function.
	
	\item[$\bullet$] We give some numerical evaluations to illustrate our results. In particular, in the common monostatic-downlink ISAC model \cite{li_computation_2025}, we observe that maximizing mutual information $I(\rv{S};\rv{T}|\rv{X})$ to design the waveform is not the optimal choice in some cases. This provides brand-new insights into the waveform design of the transmitter in ISAC systems.
\end{enumerate}

\subsection{Notation}
We use $[n] = \{ 1,...,n \}$, $\mathbb{Z}_{\ge 0} = \{ 0, 1,... \}$ to represent integer intervals. Let $1\{ \cdot \}$ denote the indicator function. For given $\set{X}$ and random variable $\rv{X} \in \set{X}$, we write $\rv{X} \sim P_{\rv{X}}$ to indicate that the random variable $\rv{X}$ follows the distribution $P_{\rv{X}}$. Let $\rv{X}^n = (\rv{X}_1,...,\rv{X}_n)$ and $x^n = (x_1,...,x_n)$ denote the random vector and its realization in the $n$-th Cartesian product $\mathcal{X}^n$, respectively. 
The probability and mathematical expectation are denoted by $\mathbb{P}[\cdot]$ and $\mathbb{E}[\cdot]$, respectively.

\section{Preliminaries}

\subsection{System Model}

Consider the SDMC model depicted in Figure \ref{sys}. In this setup, the transmitter sends the channel input codeword $\rv{X}^n$, while the state sequence $\rv{S}^n$ (drawn i.i.d. from $P_{\rv{S}}$) also acts as an input to the channel. These two sequences, $\rv{X}^n$ and $\rv{S}^n$, together pass through the channel $P_{\rv{Y} \rv{T}|\rv{X}\rv{S}}$, resulting in the outputs $\rv{Y}^n$ and $\rv{T}^n$. The state sequence $\rv{S}^n$ is perfectly known to the receiver. The transmitter aims to reliably transmit messages over the codeword $\rv{X}^n$ and to estimate $\rv{S}^n$ based on $\rv{X}^n$ and $\rv{T}^n$. The receiver reconstructs the message using $\rv{Y}^n$ and $\rv{S}^n$. 

\subsection{Code}

In the ISAC problem, an important constraint is the distortion between the source $\rv{S}^n$ and $\rv{Z}^n$. The state estimator function is $\map{h}: \set{X}^n \times \set{T}^n \mapsto \set{Z}^n $, where $i=1,...,n$. The distortion measure $\map{d}: \set{S}^n \times \set{Z}^n \mapsto [0, +\infty] $ is used to quantify the performance of the source estimator. In the spirit of the noisy source coding problem, we decompose the estimator into a pair of random mappings $\map{h}_{a}^{(K)}: \set{X}^n \times \set{T}^n \mapsto \rv{W} $ and $\map{h}_b^{(K)}: \rv{W} \mapsto \set{Z}^n$, where $\rv{W} \in \{ 1,...,K \}$. Then, Definition \ref{definition_estimator} below comes into play.
\begin{definition}\label{definition_estimator}
	The estimator $\map{h} = \map{h}_a^{(K)} \circ \map{h}_b^{(K)}$ is a $(n,K,D)$-estimator for 
	\begin{equation}
		\left\{ \set{S}^n,\set{X}^n,\set{Z}^n,  \set{T}^n, P_{\rv{S}^n}, P_{\rv{X}^n}, P_{\rv{T}^n|\rv{X^n}\rv{S^n}},\map{d} \right\} \nonumber
	\end{equation} 
	if $\lim \limits_{n \rightarrow \infty} \mathbb{E}[\map{d}(\rv{S}^n, \rv{Z}^n)]$.
\end{definition}
We adopt the standard definition of information theory. The estimation rate-distortion function of the estimator is defined as
\begin{align}
	R(D) &= \lim \limits_{n \rightarrow \infty} \frac{1}{n} \log K^\star(n,D), 
\end{align}
where
$
K^\star (n,D)  = \min \{ K: \exists(n,K,D)\text{-estimator} \}.
$

A capacity-rate-distortion code is a channel encoder function $\map{f}_i: \set{M} \times \set{T}^{i-1} \mapsto \set{X}$, a channel decoder function $\map{g}: \set{S}^n\times\set{Y}^n \mapsto \set{M} \cup \{ e \}$, where $\set{M} = \{1,...,M\}$ and  $i=1,...,n$. A cost function $\map{c}: \set{X} \mapsto [0,+\infty]$ may be imposed on the channel inputs. A distortion measure and an estimator, as defined above.

\begin{definition} \label{definition_code}
	The tuple $(\map{f},\map{g},\map{h})$ is a capacity-rate-distortion $(n,M,K,D)$-code for
	\begin{equation}
		\left\{ \set{M}, \set{S}^n,\set{X}^n, \set{T}^n, \set{Y}^n, \hat{\set{M}}, P_{\rv{S}^n}, P_{\rv{Y^n}\rv{T^n}|\rv{X^n}\rv{S^n}},\map{d},\map{c}\right\} \nonumber
	\end{equation}
	if
	\begin{align}
		\lim \limits_{n \rightarrow \infty} \mathbb{P}[\hat{\rv{M}} \neq \hat{\rv{M}} ] & = 0, \label{definition_avg_error} \\
		\lim \limits_{n \rightarrow \infty} \mathbb{E}[\map{d}(\rv{S}^n, \rv{Z}^n)] &\le D \label{definition_avg_distorition}\\
		\lim \limits_{n \rightarrow \infty} \frac{1}{n} \sum_{i=1}^{n} \mathbb{E}[\map{b}(\rv{X}_i)] & \le B. \label{definition_avg_cost}
	\end{align}
\end{definition}

\begin{definition}
	The capacity-rate-distortion tuple $(C,R,D)$ is achievable of $(n,M,K,D)$-code that simultaneously satisfy (\ref{definition_avg_error})-(\ref{definition_avg_cost}), where
	\begin{equation}
		C = \lim \limits_{n \rightarrow \infty}  \frac{1}{n} \log M, \ R = \lim \limits_{n \rightarrow \infty}  \frac{1}{n} \log K.
	\end{equation}
\end{definition}

\begin{definition}
	The capacity-rate-distortion region $\set{C}\set{R}\set{D}$ is given by the closure of the union of all achievable rate-distortion tuples $(C,R,D)$.
\end{definition}

\section{Main Result}

In this section, we present the main result. First, the rate-distortion function of the estimator is given in subsection \ref{section_rate-distortion}. In Subsection \ref{section_ba_algorithm}, we propose an improved Blahut-Arimoto-type algorithm to calculate the rate-distortion function and prove its convergence. Finally, Subsection \ref{section:capacity-rate} presents the conclusion of the capacity-rate-distortion region.

\subsection{Rate-Distortion Functions} \label{section_rate-distortion}
Our main result in this subsection is the following. 
\begin{theorem}\label{thm:rate_distortion_theorem_single}
	Fix a $P_{\rv{X}}$ and $P_{\rv{S}}$. For stationary memoryless channels $P_{\rv{T}|\rv{XS}}^{ n}$, we have
	\begin{equation}\label{eq:rate_distortion_function}
		R(D) = \inf \limits_{P_{\rv{Z}|\rv{XT}}:\,\mathbb{E}[\map{d}(\rv{S},\rv{Z})] \le D}I(\rv{T};\rv{Z}|\rv{X})
	\end{equation} 
\end{theorem}
\begin{remark}
	Theorem \ref{thm:rate_distortion_theorem_single} provides the minimum rate required for the estimator when the expected distortion is constrained to be less than $D$. Compared to similar results in \cite{chiriyath_2016_inner} and \cite{dong_rethinking_2023}, our result possesses clear operational meaning.
\end{remark}
\begin{proof}
	\underline{\textbf{Achievability part:}}
	For convenience, we use the notation $\epsilon = \mathbb{P}[\map{d}(\rv{S}^n, \map{h}(\rv{T}^n)) > d]$ to represent the excess-distortion probability.
	Upon observing $t^n \in \mathcal{T}^n$ and $x^n \in \set{X}^n$, the optimum estimator, which minimizes the excess-distortion probability, can be developed by minimizing
	\begin{equation}
		\mathbb{E}[1\left\{ \map{d}(\rv{S}^n,h(\rv{T}^n)) > d\right\}] = \mathbb{E}[\mathbb{P}[ \map{d}(\rv{S}^n,h(\rv{T}^n)) > d |\rv{X}^n,\rv{T}^n]].
	\end{equation}
	For given codewords $(z_1^n,...,z_{K}^n)$, define
	\begin{equation}
		\pi(t^n,z_i^n,x^n) = \mathbb{P}[\map{d}(\rv{S}^n,z_i^n) > d|\rv{T}^n=t^n,\rv{X}^n=x^n].
	\end{equation}
	Let
	\begin{equation}
		\map{h}_a^{(K)}(x^n,t^n) = \argmin \limits_{i \in \{1,...,K\}} \pi(t^n,z_i^n,x^n),
	\end{equation}
	where ties are broken arbitrarily. And then for $\map{h}_a^{(K)}(x^n,t^n) = i ^{\star}$, let $ \map{h}_b^{(K)} = z_{i^\star}^n$. We obtain that $\map{h} = \map{h}_a^{(K)} \circ \map{h}_b^{(K)}$ is the optimum estimator. 
	
	Assume the channel codeword $\rv{X}^n$ is drawn i.i.d. from $P_{\rv{X}}$. The estimator has full access to the transmitted channel codeword. Given $\rv{X}^n$, let the source codebook $\rv{Z}^{(K)} = (\rv{Z}_1^n,...,\rv{Z}_{K}^n)$ be drawn i.i.d. from $P_{\rv{Z}|\rv{X}}$, independently of the random variable $\rv{T}^n$, \ie, $P_{\rv{X} \rv{T} \rv{Z}^K}^n = P_{\rv{X}}^n \times P_{\rv{T}|\rv{X}}^n \times P_{\rv{Z}|\rv{X}}^n \times ... \times P_{\rv{Z}|\rv{X}}^n$. The excess-distortion probability satisfies
	\begin{align}
		\epsilon = & \mathbb{E} \left[ \min \limits_{i \in [K]} \pi ( \rv{T}^n, \rv{Z}_i^n, \rv{X}^n)\right] \nonumber \\
		= & \mathbb{E} \left[ \mathbb{E} \left[ \min \limits_{i \in [K] } \pi ( \rv{T}^n, \rv{Z}_i^n, \rv{X}^n) \Big | \rv{X}^n\right]\right] \\
		= & \mathbb{E} \left[ \int_{0}^{1} \mathbb{P}\left[ \min \limits_{i \in [K]} \pi ( \rv{T}^n, \rv{Z}_i^n, \rv{X}^n) > y \Big | \rv{X}^n \right] dy \right]. \label{thm:rc_estimator_proof_1} 
	\end{align}
	For given $x^n$, note that $\rv{Z}_1^n,...,\rv{Z}_{K}^n$ are independent random variables. We obtain that the integrand function in (\ref{thm:rc_estimator_proof_1}) can be bounded as follows:
	\begin{align}
		&\mathbb{P} \left[ \min \limits_{i \in [K]} \pi ( \rv{T}^n, \rv{Z}_i^n, x^n) > y \right]  \\
		= & \mathbb{P}\left[ \bigcap_{i \in [K]} \left\{ \pi ( \rv{T}^n, \rv{Z}_i^n, x^n) > y \right\} \right] \\
		\le & \mathbb{P}\left[ \pi ( \rv{T}^n, \rv{Z}^n, x^n) > y  \right] + \mathbb{P}\left[ i(\rv{T}^n;\rv{Z}^n|x^n) > \log K - \gamma \right] \nonumber \\
		& \ + e^{-\exp {(\gamma)}} \label{thm:rc_estimator_proof_3} ,
	\end{align}
	where (\ref{thm:rc_estimator_proof_3}) folows from the nonasymptotic covering lemma \cite[Lemma 5]{verdu_non-asymptotic_2012} and 
	\begin{equation}
		i(t^n;z^n|x^n) = \log \frac{dP_{\rv{Z}|\rv{XT}}^n(z^n|x^n,t^n)}{dP_{\rv{Z}|\rv{X}}^n(z^n|x^n)}
	\end{equation}
	is the conditional information density.
	Finally, we notice that
	\begin{align}
		&\mathbb{E}\left[ \int_{0}^{1}\mathbb{P}\left[ \pi ( \rv{T}^n, \rv{Z}^n, \rv{X}^n) > y  | \rv{X}^n  \right] dy\right] \nonumber \\
		= & \mathbb{E} \left[ \mathbb{E}\left[\pi(\rv{T}^n, \rv{Z}^n,\rv{X}^n) | \rv{X}\right] \right] \\
		= & \mathbb{P}\left[\map{d}(\rv{S}^n,\rv{Z}^n) > d\right].
	\end{align}
	The above argument implies that there exists a codebook $(z_1^n,...,z_K^n)$ such that the excess-distortion probability satisfies
	\begin{align} \label{eq:achi_bound}
		\epsilon 
		& \le \mathbb{P}\left[\map{d}(\rv{S}^n,\rv{Z}^n) > d\right] + \mathbb{P}\left[ i(\rv{T}^n;\rv{Z}^n|\rv{X}^n) > \log K - \gamma \right] \nonumber \\
		& \ \ \ \ + e^{-\exp(\gamma)}.
	\end{align}
	
	Let $\gamma = \delta n$ and let $\log K = n I(\rv{T};\rv{Z}|\rv{X}) + 2 \delta n$, where $\delta > 0$ is a constant. Take any $P_{\rv{Z}|\rv{XT}}$ such that $D = \mathbb{E}\left[ \map{d}(\rv{S},\rv{Z}) \right]$ and let $d= D -\delta < D$. Apply (\ref{eq:achi_bound}) to conclude that there exists a estimator $\map{h}:\set{T}^n \mapsto \set{Z}^n$ such that 
	\begin{align}\label{achievability_asy_proof_1}
		& \mathbb{P}\left[\map{d}(\rv{S}^n,\map{h}(\rv{T}^n)) > d\right] \nonumber \\
		\le & \mathbb{P}\left[\map{d}(\rv{T}^n,\rv{Z}^n) > d \right] + \mathbb{P}\left[ i(\rv{T};\rv{Z}|\rv{X}) > n( I(\rv{T};\rv{Z}|\rv{X})+\delta) \right] \nonumber \\
		&  + e^{-\exp(\delta n)} 
	\end{align}
	We bound the three terms in (\ref{achievability_asy_proof_1}), respectively. Note that 
	\begin{equation}
		i(\rv{T}^n;\rv{Z}^n|\rv{X}^n) = \sum_{i=1}^{n} \log \frac{dP_{\rv{Z}|\rv{XT}}}{dP_{\rv{Z}|\rv{X}}}(\rv{T}_i,\rv{Z}_i,\rv{X}_i)
	\end{equation}
	is a sum of i.i.d. random variables with mean $I(\rv{T};\rv{Z}|\rv{X})$. By WLLN, we obtain that
	$
	\lim \limits_{n \rightarrow \infty} \mathbb{P}\left[i(\rv{T}^n;\rv{Z}^n|\rv{X}^n) > n( I(\rv{T};\rv{Z}|\rv{X}) + \delta )\right] \rightarrow 0 .
	$
	On the other hand, we also have
	$
	\lim \limits_{n\rightarrow \infty}\mathbb{P}\left[\map{d}(\rv{S}^n,\rv{Z}^n) > D'\right] \rightarrow 0
	$
	and $\lim \limits_{n\rightarrow \infty}e^{-\exp(\delta n)} \rightarrow 0$. We combine the above argument to obtain that the RHS of (\ref{achievability_asy_proof_1}) tends to zero. Then, note that the average distortion can be bounded by 
	\begin{align} 
		&\mathbb{E}\left[ \map{d}(\rv{S}^n,\map{h}(\rv{T}^n)) \right] \nonumber \\
		\le & D - \delta + \mathbb{E}\left[ \map{d}(\rv{S}^n,\map{h}(\rv{T}^n)) 1 \left\{ \map{d}(\rv{S}^n,\map{h}(\rv{T}^n)) > D' \right\} \right] \label{achievability_asy_proof_2} \\
		= & D - \delta + d_{\max}\mathbb{P}\left[\map{d}(\rv{S}^n,\map{h}(\rv{T}^n)) > D'\right] \\
		= & D - \delta +o(1) \\
		\le & D
	\end{align}
	Thus, for sufficiently large $n$, the expected distortion is at most $D$, as required.
	
	We have shown that there exists a $(n,K,D)$-estimator with $\log K = n I(\rv{T};\rv{Z}|\rv{X}) + 2 \delta n$ and hence we obtain that
	\begin{equation}
		R(D) = \inf \limits_{P_{\rv{Z}|\rv{XT}}:\mathbb{E}[\map{d}(\rv{S},\rv{Z})] \le D}I(\rv{T};\rv{Z}|\rv{X})
	\end{equation}
	is achievable.
	
	{\noindent \underline{\textbf{Converse part:}}	Let } 
	\begin{equation}
		j(t|zx) = \log \frac{dP_{\rv{T}|\rv{ZX}}^n(t^n|z^n,x^n)}{dP_{T|X}^n(t^n|x^n)}.
	\end{equation}
	Define the encoder and decoder to be $P_{\rv{K}|\rv{XT}}^n$ and $P_{\rv{Z}|\rv{K}}^n$, where $\rv{K}\in [K]$. 
	We have
	\begin{align}
		&\mathbb{P}\left[ j(\rv{T}^n|\rv{Z}^n\rv{X}^n) > \log K + \gamma \right] \nonumber \\
		= & \mathbb{P}\left[ j(\rv{T}^n|\rv{Z}^n\rv{X}^n) > \log K + \gamma, \map{d}(\rv{S}^n,\rv{Z}^n) > d \right] \nonumber \\
		& +  \mathbb{P}\left[ j(\rv{T}^n|\rv{Z}^n\rv{X}^n) > \log K + \gamma, \map{d}(\rv{S}^n,\rv{Z}^n) \le d \right] \\
		\le & \epsilon + \mathbb{P}\left[ j(\rv{T}^n|\rv{Z}^n\rv{X}^n) > \log K + \gamma \right] \\
		\le & \epsilon + \frac{\exp(-\gamma)}{K}\mathbb{E}\left[ \exp\left( j(\rv{T}^n|\rv{Z}^n\rv{X}^n) \right) \right] \label{thm:converse_proof_1}\\
		\le & \epsilon + \frac{\exp(-\gamma)}{K} \sum_{i=1}^{K} \int_{x^n \in \set{X}^n}dP_{\rv{X}}^n(x^n) \nonumber \\ 
		& \ \ \int_{z^n \in \set{Z}^n} dP_{\rv{Z}|\rv{K}}^n(z^n|k^n) \int_{t^n \in \set{T}^n} dP_{\rv{T}|\rv{XZ}}^n(t^n|x^n,z^n)\label{thm:converse_proof_2}\\
		= & \epsilon + \exp(-\gamma), \label{thm:converse_proof_3}
	\end{align}
	where (\ref{thm:converse_proof_1}) follows from Markov's inequality, (\ref{thm:converse_proof_2}) simply expands the expectation and noting $P_{\rv{K}|\rv{XT}}(k|x,t) \le 1$ for every $(x,t,k) \in \set{X} \times \set{T} \times \{ 1,...,K \}$. 
	
	Then, we prove the converse part by contradiction. Suppose there exists a $(n,M,D)$-estimator such that $\mathbb{E}[\map{d}(\rv{S}^n,\rv{Z}^n)] \le D$ but $R(D) < I(\rv{T;\rv{Z}|\rv{X}})$. 
	Fix any $\delta > 0$. Let $\log K = n I(\rv{T;\rv{Z}|\rv{X}}) - 2\delta n < n I(\rv{T;\rv{Z}|\rv{X}}) $ and $\gamma = \delta n$. By virtue of (\ref{thm:converse_proof_3}), we get for any $(n,M,D)$-estimator:
	\begin{align}
		&\mathbb{P}\left[\map{d}(\rv{S}^n,\map{h}(\rv{T}^n)) > D\right] \nonumber  \\
		\ge & \mathbb{P}[i(\rv{T}^n;\rv{Z}^n|\rv{X}^n) \ge n (I(\rv{T;\rv{Z}|\rv{X}}) - \delta)] - \exp(-\gamma) \label{converse_asy_proof_1}.
	\end{align}
	However, the probability on the RHS of (\ref{converse_asy_proof_1}) tends to $1$ and thus $\mathbb{E}[\map{d}(\rv{S}^n,\map{h}(\rv{T}^n))] > D$, thereby contradicting the fact that $\mathbb{E}\left[ \map{d}(\rv{S}^n,\map{h}(\rv{T}^n)) \right] \le D$. Hence any $(n,M,D)$-estimator such that $\mathbb{E}\left[ \map{d}(\rv{S}^n,\map{h}(\rv{T}^n)) \right] \le D$ should satisfy $\log K \ge I(\rv{T;\rv{Z}|\rv{X}})$.
	
	\end{proof}

\subsection{Computation of Rate-Distortion Functions} \label{section_ba_algorithm}

In this subsection, we present a Blahut-Arimoto type algorithm that can be used to solve the optimization problem (\ref{eq:rate_distortion_function}).
By using the same argument in the proof of \cite[Theorem 2.7.4]{cover_elements_nodate}, the mutual information $I(\rv{T};\rv{Z}|\rv{X})$ is a convex function of $P_{\rv{Z}|\rv{XT}}$ for fixed $P_{\rv{X}}$ and $P_{\rv{S}}$. To solve (\ref{eq:rate_distortion_function}), we have the following optimization problem:
\begin{align}
	L_\mu 
	& = \min \limits_{P_{\rv{Z}|\rv{XT}}} I(\rv{T};\rv{Z}|\rv{X}) - \mu \mathbb{E}[\map{d}(\rv{S},\rv{Z})].
\end{align}
Define
\begin{equation}
	F_\mu(P_{\rv{Z}|\rv{XT}}, Q_{\rv{Z}|\rv{X}} ) = D(P_{\rv{Z}|\rv{XT}} \| Q_{\rv{Z}|\rv{X}} |P_{\rv{XT}}) - \mu \mathbb{E}[\map{d}(\rv{S},\rv{Z})].
\end{equation}
Then, we consider the Blahut-Arimoto type alternating optimization techniques \cite{blahut_computation_1972,csiszar_computation_1974}. 
\begin{lemma} \label{thm:alternating_minimization}
	For given $P_{\rv{X}}$ and $P_{\rv{S}}$, we have
	\begin{enumerate}
		\item[(1)] For fixed $P_{\rv{Z}|\rv{XT}}$, $F_\mu(P_{\rv{Z}|\rv{XT}}, Q_{\rv{Z}|\rv{X}} )$ is minimized by $Q_{\rv{Z}|\rv{X}}(z|x) = Q_{z|x}(P_{\rv{Z}|\rv{XT}})$ for $(x,z) \in \set{X} \times \set{Z}$, where 
		\begin{equation}
			Q_{z|x}(P_{\rv{Z}|\rv{XT}}) = \sum_{t \in \set{T}} \sum_{s \in \set{S}} P_{\rv{S}}(s) P_{\rv{T}|\rv{XS}} (t|x,s) P_{\rv{Z}|\rv{XT}}(z|x,t). \label{eq:minimization_Q}
		\end{equation}
		\item[(2)] For fixed $Q_{\rv{Z}|\rv{X}}$, $F_\mu(P_{\rv{Z}|\rv{XT}}, Q_{\rv{Z}|\rv{X}} )$ is minimized by $P_{\rv{Z}|\rv{XT}}(z|xt) = P_{z|xt}(Q_{\rv{Z}|\rv{X}})$ for $(x,z,t) \in \set{X} \times \set{Z} \times \set{T}$, where
		\begin{equation}
			P_{z|xt}(Q_{\rv{Z}|\rv{X}}) = \frac{Q_{\rv{Z}|\rv{X}}(z|x)\exp \left( \mu \mathbb{E}[\map{d}(\rv{S},z)|x,t] \right)}{\sum_{a \in \set{Z}}Q_{\rv{Z}|\rv{X}}(a|x)\exp \left( \mu \mathbb{E}[\map{d}(\rv{S},a)|x,t] \right)}, \label{eq:minimization_P}
		\end{equation}
		and $\mu \le 0$ specifies the value of $D$ and $R(D)$.
	\end{enumerate}
\end{lemma}
\begin{proof}
	See Appendix \ref{Appendix_BA_algorithm}.
\end{proof}

The following results show that the alternating optimization using Theorem \ref{thm:alternating_minimization} converges to the rate-distortion function that satisfies some target distortions. In the proof, we use the techniques similar to \cite{Arimoto_1972,csiszar_computation_1974}.

\begin{theorem}\label{thm:converge_to_rate_distortion_function}
	Suppose alphabet $\set{X}, \set{Z}, \set{T}$ are finite and start from an arbitrary probability measure $Q^{(0)}$ with $Q^{(0)}(z|x)>0$. Let $P_{z|xt}(\cdot)$ and $Q_{z|x}(\cdot)$ be defined as (\ref{eq:minimization_P}) and (\ref{eq:minimization_Q}), respectively. For $k \in \mathbb{Z}_{\ge 0} $, set recursively 
	\begin{align}
		P^{(k+1)}(z|xt) & = P_{z|xt} \left( Q^{(k)} \right), \label{eq:alternating_P} \\
		Q^{(k+1)}(z|x) &= Q_{z|x} \left( P^{(k+1)}\right),\label{eq:alternating_Q}
	\end{align} 
	Then, there exists a $Q_{\rv{Z}|\rv{X}}^\star$ such that $Q^{(k)} \rightarrow Q_{\rv{Z}|\rv{X}}^\star, \ P^{(k)} \rightarrow P_{\rv{Z}|\rv{XT}}^{\star} = P(Q_{\rv{Z}|\rv{X}}^\star)$, and $F_\mu(P_{\rv{Z}|\rv{XT}}^{\star},Q_{\rv{Z}|\rv{X}}^\star) = L_\mu$.
\end{theorem}
\begin{proof}
	See Appendix \ref{Appendix_BA_algorithm}.
\end{proof}

Through the recursive procedure established in Theorem \ref{thm:converge_to_rate_distortion_function}, we obtain the optimal estimator $P_{\rv{Z}|\rv{XT}}$. We describe this procedure  using Algorithm \ref{Algorithm:BA_SDRD}.

\begin{algorithm}[h]
	\caption{Blahut-Arimoto Type Algorithm for $R(D)$}
	\label{Algorithm:BA_SDRD}
	\begin{algorithmic}[1]
		\REQUIRE
		$\mu \le 0$, $\delta > 0$
		\STATE Set $k=0$ and let $Q^{(0)}_{\rv{Z}|\rv{X}}(z|x) = \frac{1}{|\set{Z}|}$ for each $x \in \set{X}$;
		\REPEAT
		\STATE $k \leftarrow k+1$;
		\STATE Compute $P^{(k)}_{\rv{Z}|\rv{XT}}$ using (\ref{eq:alternating_P});
		\STATE Compute $Q^{(k)}_{\rv{Z}|\rv{X}}$ using (\ref{eq:alternating_Q});
		\UNTIL $\sum_{x \in \set{X}} \sum_{t \in \set{T}} \left| Q^{(k)}_{\rv{Z}|\rv{X}}(z|x)-Q^{(k-1)}_{\rv{Z}|\rv{X}}(z|x) \right| \le \delta$.
		\STATE Set $P_{\rv{Z}|\rv{XT}} \leftarrow P^{(k)}_{\rv{Z}|\rv{XT}}$.
	\end{algorithmic}	
	
\end{algorithm}

\subsection{Capacity-Rate-Distortion Region} \label{section:capacity-rate}
As shown in \cite{ahmadipour_information-theoretic_2024}, the following lemma achieves the minimum distortion. 
\begin{lemma}[\textit{\texorpdfstring{\cite[Lemma 1]{ahmadipour_information-theoretic_2024}}}] \label{lemma_minimum_distortion}
	Let $\map{c}(x) = \mathbb{E}[\pi(x,\rv{T})|\rv{X}=x],$
	\begin{equation} \label{eq:minimum_lamma_1}
		\pi(x,t) = \min \limits_{z \in \set{Z}} \sum_{s \in \set{S}} P_{\rv{S}|\rv{XT}}(s|x,t) \map{d}(s,z).
	\end{equation}
	The minimum distortion $\mathbb{E}[\map{c}(\rv{X})]$ is achieved by
	\begin{equation} \label{eq:deterministic_estimator}
		\map{h}^{\star}(z|x,t) = 
		\begin{cases}
			1 & , z = \argmin\limits_{z \in \set{Z}}  \sum_{s \in \set{S}} P_{\rv{S}|\rv{XT}}(s|x,t) \map{d}(s,z) \\
			0 & , o.w.
		\end{cases}
	\end{equation}
\end{lemma}

We define the following set:
\begin{align}
	\set{P}^{\star}(D,B) = \left\{ P_{\rv{X}} \, \middle | \, 
		\begin{aligned}
			\mathbb{E}[\map{c}(\rv{X})] &\le D, \\
			\mathbb{E}[\map{b}(\rv{X})] &\le B.
		\end{aligned} \right\} 
\end{align}
Then, the capacity-rate-distortion region can be bounded as following result.
\begin{proposition}
	The capacity-rate-distortion region $\set{C}\set{R}\set{D}$ satisfies the constraints
	\begin{align}
		C & \le \max \limits_{P_{\rv{X}} \in \set{P}^{\star}(D_0,B)} I(\rv{X};\rv{Y}|\rv{S}), \\
		R & \ge \min_{P_{\rv{Z|\rv{XT}}}:\mathbb{E}_{P_{\rv{X}}^{\star}P_{\rv{Z|\rv{XT}}}}[\map{d}(\rv{S},\rv{Z})]\le D} I(\rv{T};\rv{Z}|\rv{X}), \\
		D & \ge \mathbb{E}_{P_{\rv{X}}^{\star}P_{\rv{Z|\rv{XT}}}^{\star}}[\map{d}(\rv{S},\rv{Z})],
	\end{align}
	where $D_0 \in \{D \,| \, D \ge \min \limits_{P_{\rv{X}}} \mathbb{E}[\map{c}(\rv{X})] \}$, 
	\begin{align}
		P_{\rv{X}}^{\star} &= \argmax \limits_{P_{\rv{X}} \in \set{P}^{\star}(D_0,B)} I(\rv{X};\rv{Y}|\rv{S}), \\
		P_{\rv{Z|\rv{XT}}}^{\star} &= \argmin \limits_{P_{\rv{Z|\rv{XT}}}:\mathbb{E}_{P_{\rv{X}}^{\star}P_{\rv{Z|\rv{XT}}}}[\map{d}(\rv{S},\rv{Z})]\le D} I(\rv{T};\rv{Z}|\rv{X}).
	\end{align}
\end{proposition}
\begin{proof}
	Combine Theorem 1 in \cite{ahmadipour_information-theoretic_2024} and Theorem \ref{thm:rate_distortion_theorem_single}, we conclude the result. 
\end{proof}

\begin{figure*}[t]
	\normalsize	
	\centering
	\begin{minipage}[t]{0.4\linewidth}
		\centering
		\includegraphics[width = 1\textwidth]{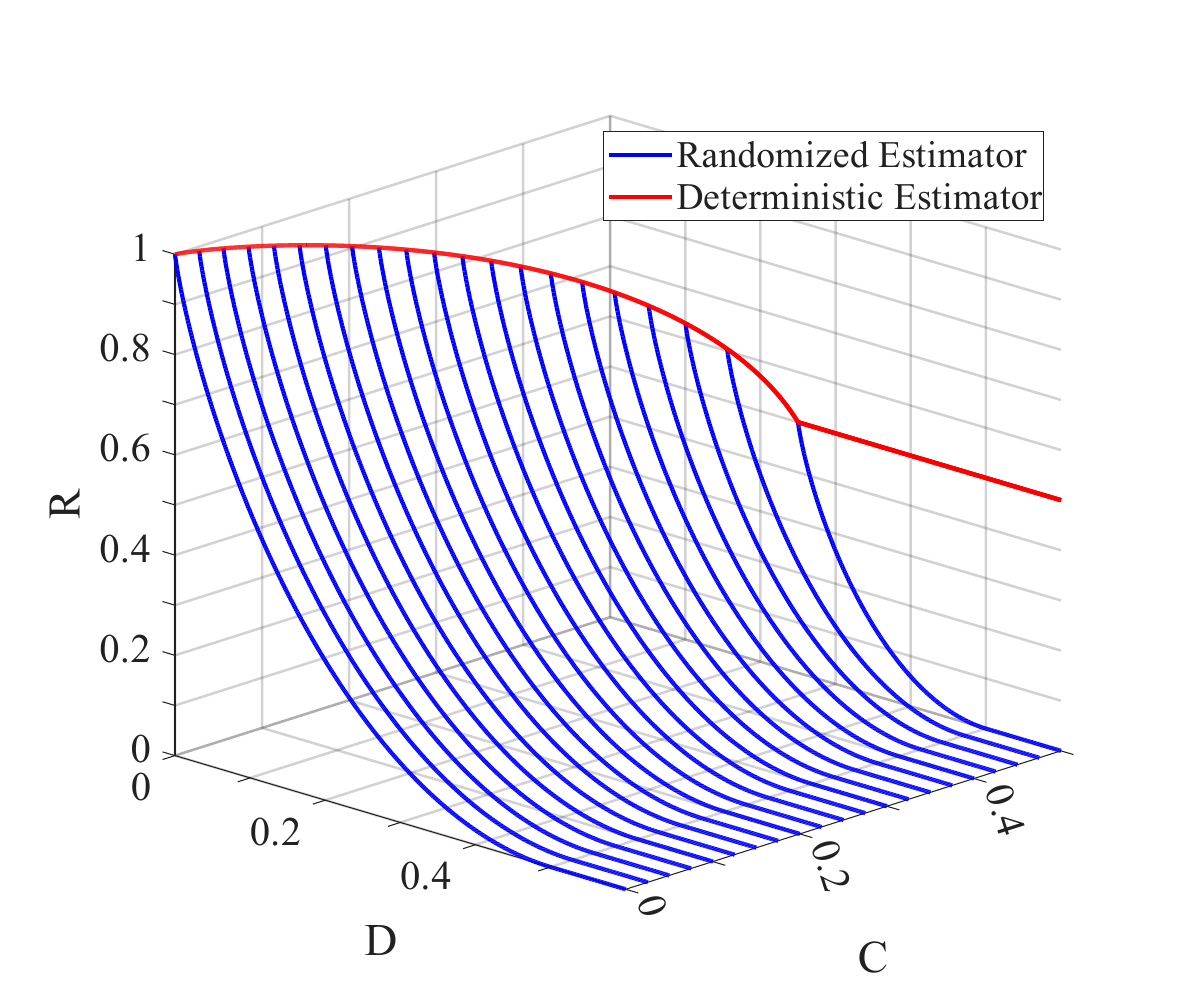}
		\caption{\ Capacity-distortion tradeoff of the binary channel with multiplicative
			bernoulli state with $q=0.5$.}
		\label{Fig:binary1}
	\end{minipage}
	\hspace{40pt}
	\begin{minipage}[t]{0.4\linewidth}
		\centering
		\includegraphics[width = 1\textwidth]{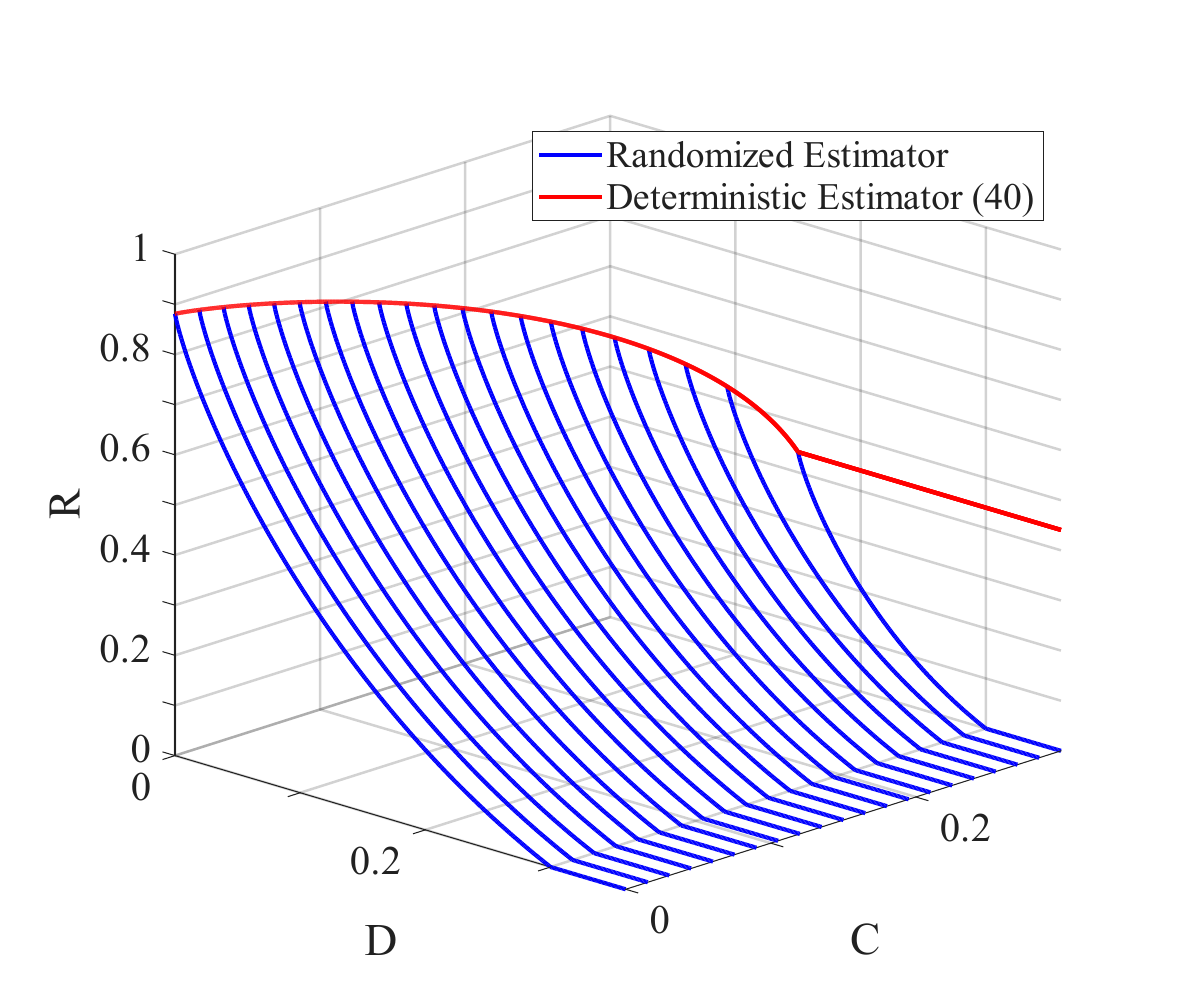}
		\caption{\ Capacity-distortion tradeoff of the binary channel with multiplicative
			bernoulli state with $q=0.3$.}
		\label{Fig:binary2}
	\end{minipage}
\end{figure*}

\begin{figure*}[t]
	\normalsize	
	\centering
	\begin{minipage}[t]{0.4\linewidth}
		\centering
		\includegraphics[width = 1\textwidth]{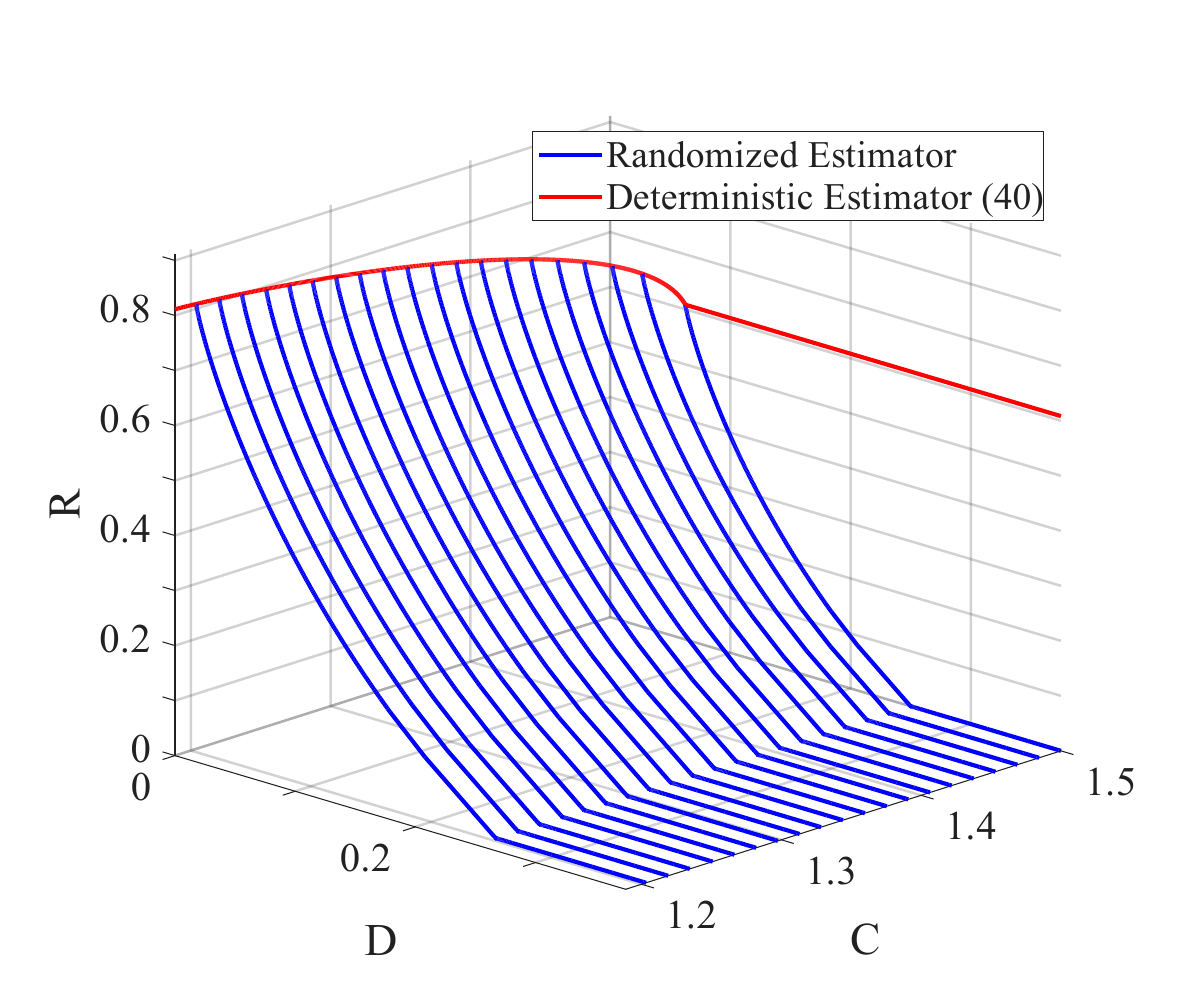}
		\caption{\ Capacity-distortion tradeoff of the DMC with multiplicative state with $P_{\rv{S}}=[1/4,1/4,1/4,1/4]$.}
		\label{Fig:DMC1}
	\end{minipage}
	\hspace{40pt}
	\begin{minipage}[t]{0.4\linewidth}
		\centering
		\includegraphics[width = 1\textwidth]{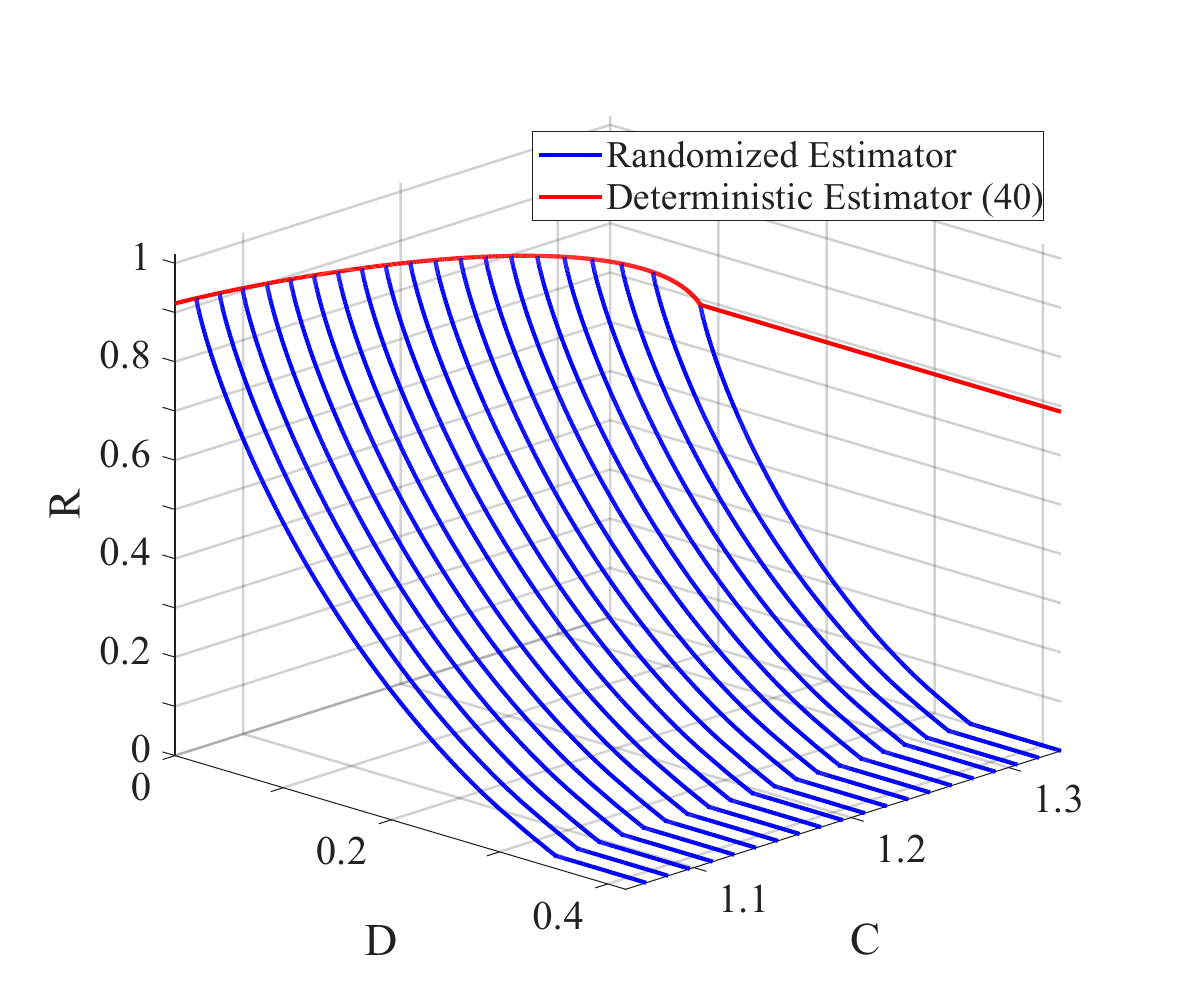}
		\caption{\ Capacity-distortion tradeoff of the DMC with multiplicative state with $P_{\rv{S}}=[1/3,1/4,1/4,1/6]$.}
		\label{Fig:DMC2}
	\end{minipage}
\end{figure*}

\section{Numerical Results}
In this section, we give some examples to intuitively demonstrate our results, including the capacity-rate-distortion region for binary and general discrete memoryless channels (DMCs) with multiplicative state, as well as the rate-distortion function in a common monostatic-downlink ISAC model.
\subsection{Binary Channels with Multiplicative Bernoulli States}
 We consider the binary channel with multiplicative Bernoulli state, \ie, the  binary alphabets $\set{X} = \set{S} = \set{Y} = \{0,1\}$ where the state $\rv{S}$ is Bernoulli-$q$, the channel $Y=SX$, and the Hamming distortion measure $\map{d}(\rv{S},\rv{Z}) = 1\left\{ \rv{S} \neq \rv{Z} \right\}$. The receiver gets the perfect feedback, \ie, $\rv{T}=\rv{Y}$.
	
	\begin{corollary} \label{corollary:binary_multiplicate_tradeoff}
		The capacity-rate-distortion region of a binary channel with multiplicative Bernoulli state is given by
		\begin{equation}
			\set{CRD} = \left\{ C, R,D \right\},
		\end{equation}
		where $p \in [0,1/2]$, $C \le q H_b(p)$, $p \min\{ q,1-q \} < D \le \min \{ q,1-q \}$ and
		\begin{equation} \label{BMC_randomized}
			R \ge (1-p) \left( H_b(q) - H_b\left( \frac{D-p \min \{ q,1-q \}}{1-p} \right) \right).
		\end{equation}
	\end{corollary}
	\begin{proof}
		By setting $P_{\rv{X}}(0) = p$, the channel capacity is
		\begin{equation}
			I(\rv{X};\rv{Y}|\rv{S}) = H(\rv{T}|\rv{S}) - H(\rv{T}|\rv{XS}) = qH_b(p).
		\end{equation}
		Then, we consider the rate-distortion function. First, note that if $\rv{X}=0$, the optimal estimator chooses $\argmin_{z \in \set{Z}} \sum_{s \in \set{S}} P_{\rv{S}}(s)\map{d}(s,z)$ since we always obtain $\rv{T}=0$. Let the conditional distortion be $D_0 = \mathbb{E}[\map{d}(\rv{S},\rv{Z})|\rv{X}=0]$. After some algebra, we have $D_0 = p \min \{1,1-q\}$. Next, if $\rv{X}=1$, we have $\rv{T}=\rv{S}$, so the conditional mutual information is equal to $I(\rv{S};\rv{Z})$, and $\mathbb{E}[\map{d}(\rv{S},\rv{Z})] \le D_1$. Set $D \ge D_0+D_1$. Combining the results from both cases, we obtain that
		\begin{align}
			& I(\rv{T};\rv{Z}|\rv{X})\nonumber \\
			= & (1-p) \left( H_b(q) - H_b\left( \frac{D-p \min \{ q,1-q \}}{1-p} \right) \right),
		\end{align}
		where $p \min\{ q,1-q \} < D \le \min \{ q,1-q \} $.
		This completes the proof.
	\end{proof}
	The numerical results of Corollary \ref{corollary:binary_multiplicate_tradeoff} are illustrated in Fig. \ref{Fig:binary1} and \ref{Fig:binary2} with parameters $p=0.3$ and $p = 0.3$, respectively. The randomized estimator achieves the boundary of $\set{CRD}$ by using Corollary \ref{corollary:binary_multiplicate_tradeoff}. We also use Algorithm \ref{Algorithm:BA_SDRD} to obtain the numerical results. Unsurprisingly, they coincide with Corollary \ref{corollary:binary_multiplicate_tradeoff}. Through figures, we find that the deterministic estimator (\ref{eq:deterministic_estimator}), which achieves minimum distortion, is a special case of the entire tradeoff surface. In fact, when $P_{\rv{X}}$ or $P_{\rv{S}}$ is fixed, this tradeoff problem naturally degenerates into a two-dimensional tradeoff between the target distortion $D$ and the rate-distortion function $R(D)$ (see Theorem \ref{thm:rate_distortion_theorem_single}) or the channel capacity $C(D)$ (see \cite[Section II]{ahmadipour_information-theoretic_2024}), respectively. 
	
\subsection{General DMCs with Multiplicative States}	
	Then, we consider a general DMC with a uniform multiplicative state. We also use the channel $\rv{Y} = \rv{S} \rv{X}$ and $\rv{T} = \rv{S} \rv{X}$, where $\set{X} = \set{S} = \{ 0,1,2,3 \}$. Let $\rv{S} \sim P_{\rv{S}}$ be a discrete memoryless source with Hamming distortion. The optimal input distribution is unique. We apply \cite[Theorem 4]{ahmadipour_information-theoretic_2024} to solve for the input distribution $P_{\rv{X}}$. Then, for such a $P_{\rv{X}}$, we use Algorithm \ref{Algorithm:BA_SDRD} to find the optimal estimator. The numerical results of the boundary of $\mathcal{CRD}$ are illustrated in Fig. \ref{Fig:DMC1} and \ref{Fig:DMC2}. Moreover, we once again observe that the randomized estimator obtained by Algorithm \ref{Algorithm:BA_SDRD} coincides with the deterministic estimator that achieves the minimum distortion at the minimum target distortion. 

\subsection{Real Gaussian Channels with Sensing Echo Signals Reflected From the Target}

We consider a monostatic-downlink ISAC model \cite{li_computation_2025}: the communication channel is an additive white Gaussian noise channel, and the sensing signal is an echo signal reflected from the target, as defined below.

\begin{align}
	\rv{Y}_i & = \rv{X}_i + \rv{N}_i, \\
	\rv{T}_i & = \rv{S}_i \rv{X}_i + \rv{N}_i, \label{eq:sensing_gaussian_channel}
\end{align}
where both sequence $\{\rv{N}_i\}$ and $\{\rv{S}_i\}$ are independent of each other and i.i.d. Gaussian with zero mean and variance $\sigma_\rv{N}^2$ and $\sigma_{\rv{S}}^2$, respectively; $\rv{X}_i$ is the channel input satisfying $\lim \limits_{n \rightarrow \infty} \frac{1}{n} \sum_{i=1}^{n} \mathbb{E}[\rv{X}_i^2] \le P$. The distortion measure is the quadratic distortion measure $\map{d}(s,z) = (s-z)^2$. 

Fix $\rv{X}^n =x^n$ is a deterministic waveform (uncoded), the rate distortion function can be expressed as (\eg, see \cite{dobrushin_information_1962} and \cite{wolf_transmission_1970})
\begin{equation} \label{eq:Det_rate-distortion}
	R_{x}(D) = 
	\begin{cases}
		\frac{1}{2} \log \left( \frac{x^4\sigma_{\rv{S}}^4}{D(x^2 \sigma_{\rv{S}}^2 + \sigma_{\rv{N}}^2) - x^2\sigma_{\rv{S}}^2\sigma_{\rv{N}}^2} \right), & \, D\in \set{D}, \\
		0, & \, o.w,
	\end{cases}
\end{equation}
where $\set{D} = \big (\frac{x^2\sigma_{\rv{S}}^2\sigma_{\rv{N}}^2}{x^2\sigma_{\rv{S}}^2 + \sigma_{\rv{N}}^2}, x^2\sigma_{\rv{S}}^2 \big ]$. Then, if $\rv{X}^n$  is drawn i.i.d. from $P_{\rv{X}}$ and $P_{\rv{X}}$ is a uniform distribution, we have
\begin{equation} \label{eq:Gaussian_rate-distortion}
	R(D) = \mathbb{E}[R_{\rv{X}}(D)],
\end{equation}
where $D \ge \max_{x \in \set{X}} \frac{x^2\sigma_{\rv{S}}^2\sigma_{\rv{N}}^2}{x^2\sigma_{\rv{S}}^2 + \sigma_{\rv{N}}^2}$.

We consider the deterministic and $16$-ary PAM transmitted waveform with parameters $\sigma_\rv{N}^2 = \sigma_{\rv{S}}^2 = 1$, $10 \log_{10} P$ taking $0$ and $10$ dB. For the deterministic transmitted waveform, let $\rv{X}^n = (\sqrt{P}, ..., \sqrt{P})$ to satisfy the power constraint. Let $P_{\rv{X}}$ be uniform on alphabet $\set{X}$, and the $M = 16$-ary PAM constellation is quantized to 
\begin{equation}
	\set{X} = \{ (2m-1-M)k,m=1,...,M \},
\end{equation}
where we let $k = \sqrt{3P/(M^2-1)}$ to satisfy the power constraint. The resulting rate-distortion curve is shown in Fig. \ref{fig_gaussian}. We can see that deterministic waveforms require a lower estimation rate to achieve the same expected distortion compared to 16-ary PAM waveforms at $0$ dB power constraint. However, if power constraint is $10$ dB, 16-ary PAM waveforms enable the estimator to exhibit better performance when $D \le 6$.

 \begin{figure}[t]
	\normalsize
	\includegraphics[width = 0.4\textwidth]{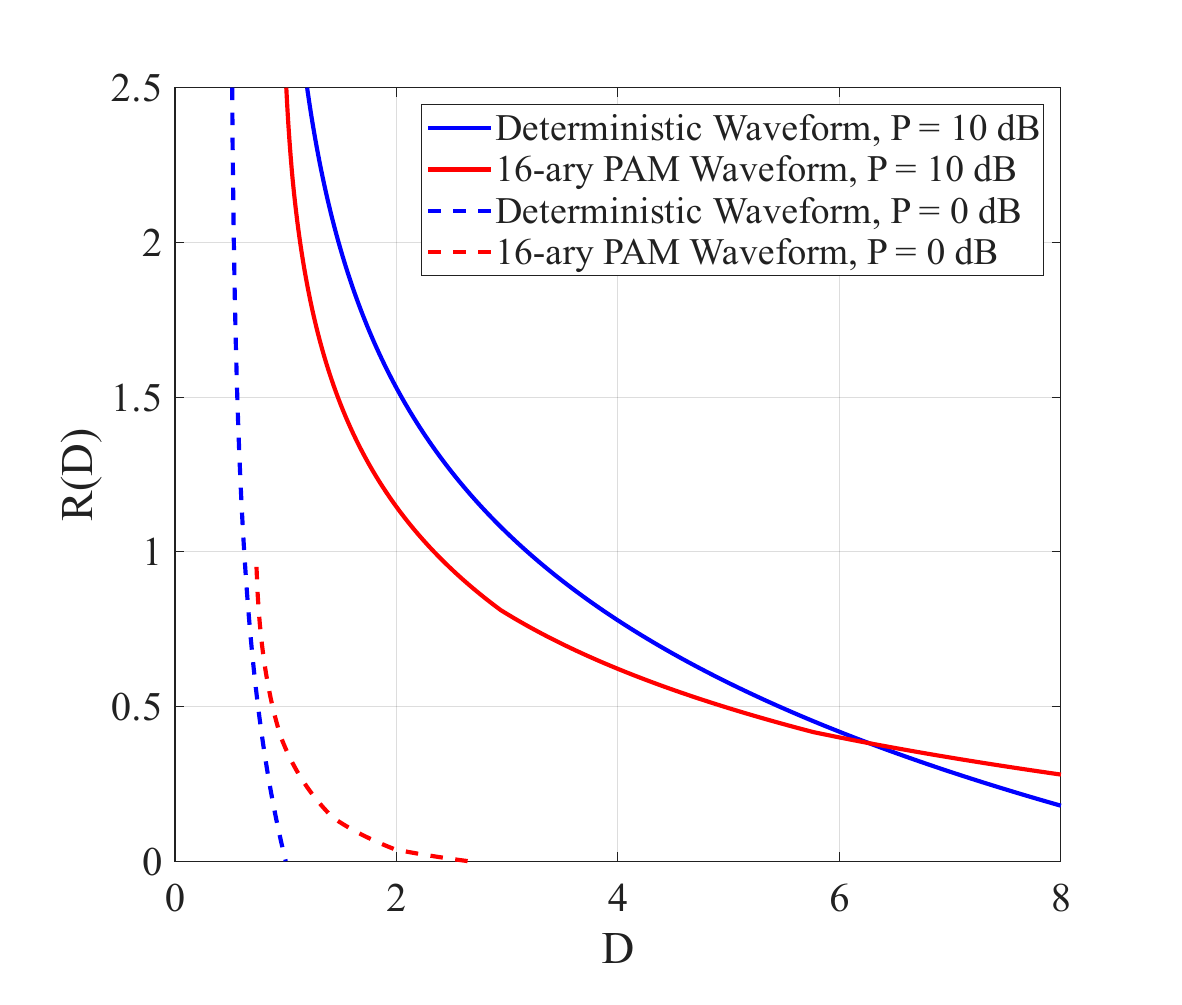}
	\caption{\ Rate-distortion function of real Gaussian channels: Deterministic waveform and $16$-ary PAM waveform.}
	\label{fig_gaussian}
\end{figure}

\begin{remark}
	It should be noted that in \cite{bell_information_1993} and \cite{yang_mimo_2007}, the authors maximize the conditional mutual information $I(\rv{S};\rv{T}|\rv{X})$ to design the deterministic waveforms. For the sensing model of (\ref{eq:sensing_gaussian_channel}), we have $I(\rv{S};\rv{T}|\rv{X}) = \mathbb{E}[\frac{1}{2} \log \left( 1+\rv{X}^2\sigma_{\rv{S}}^2/\sigma_{\rv{N}}^2 \right)] \le \log (1+ \sigma_{\rv{S}}^2/\sigma_{\rv{N}}^2 \mathbb{E}[\rv{X}^2]) = \log (1+ \sigma_{\rv{S}}^2/\sigma_{\rv{N}}^2 P)$ by the concavity of $\log(1+x)$ and Jensen's inequality, with equality if and only if $\rv{X} = \sqrt{P}$ is a constant or $P_{\rv{X}}(-\sqrt{P}) = P_{\rv{X}}(\sqrt{P})=0.5$. In numerical results of Fig. \ref{fig_gaussian}, we use $\rv{X} = \sqrt{P}$. We note that at high power constraint (\eg, $10$ dB), the $16$-ary PAM waveform gives a lower estimation rate compared to the deterministic waveform under the same expected distortion constraint when $D \le 6$. This indicates that the performance of the estimator can be improved through the coding of the transmitter, especially when compared with the waveform which maximizes $I(\rv{S};\rv{T}|\rv{X})$.
\end{remark}
\section{Conclusion and Discussion}\label{Section_conclusion}
	
In summary, we introduced an estimator rate-distortion function for the ISAC model, which jointly optimizes wireless communication and sensing tasks. This function, which characterizes the minimum rate required to achieve a given estimation distortion, has a clear operational meaning and provides information-theoretic insights into waveform design. In addition, we proposed an improved Blahut-Arimoto type algorithm, which is proven to converge to the rate-distortion function in this paper and is used to find the probability distribution of the optimal estimator. Finally, we establish a capacity-rate-distortion region that describes the achievable tradeoff between communication capacity and estimation rate, enabling us to balance communication and sensing performance in a collaborative framework.

\appendices

\section{Blahut-Arimoto Iterative algorithm} \label{Appendix_BA_algorithm}

\subsection{Proof of Theorem \ref{thm:alternating_minimization}}
On the one hand, note that
\begin{align}
	I(\rv{T};\rv{Z}|\rv{X}) 
	& = D(P_{\rv{Z}|\rv{XT}}\| P_{\rv{Z}|\rv{X}}|P_{\rv{XT}}) \label{Appendix_BA_proof_1} \\
	& = \mathbb{E}\left[ \mathbb{E}\left[\log \frac{P_{\rv{Z}|\rv{XT}}}{P_{\rv{Z}|\rv{X}} } \Big | \rv{X},\rv{T}\right] \right] \label{Appendix_BA_proof_2}  \\
	& = \mathbb{E}\left[ \mathbb{E}\left[\log \frac{P_{\rv{Z}|\rv{XT}}Q_{\rv{Z}|\rv{X}}}{P_{\rv{Z}|\rv{X}} Q_{\rv{Z}|\rv{X}}} \Big | \rv{X},\rv{T}\right] \right] \\
	& = \mathbb{E}\left[ \mathbb{E}\left[\log \frac{P_{\rv{Z}|\rv{XT}}}{Q_{\rv{Z}|\rv{X}}} \Big | \rv{X},\rv{T}\right] \right]  - \mathbb{E}\left[\log \frac{P_{\rv{Z}|\rv{X}}}{Q_{\rv{Z}|\rv{X}}} \right] \\
	& = D(P_{\rv{Z}|\rv{XT}} \| Q_{\rv{Z}|\rv{X}} |P_{\rv{XT}}) - D(P_{\rv{Z}|\rv{X}}\|Q_{\rv{Z}|\rv{X}}|P_{\rv{X}})  \\
	& \le D(P_{\rv{Z}|\rv{XT}} \| Q_{\rv{Z}|\rv{X}} |P_{\rv{XT}}) \label{Appendix_BA_proof_3}
\end{align}
where the expectation is with respect to $P_{\rv{XTZ}} = P_{\rv{X}} P_{\rv{T}|\rv{X}} P_{\rv{Z}|\rv{XT}}$, (\ref{Appendix_BA_proof_1}) follows from the definition of KL divergence, (\ref{Appendix_BA_proof_2}) we condition on $(\rv{X},\rv{T})$, (\ref{Appendix_BA_proof_3}) follows from the non-negativity of KL divergence.

On the other hand, define
\begin{align}
	&F_\mu(P_{\rv{Z}|\rv{XT}}, Q_{\rv{Z}|\rv{X}} ) = D(P_{\rv{Z}|\rv{XT}} \| Q_{\rv{Z}|\rv{X}} |P_{\rv{XT}}) - \mu \mathbb{E}[\map{d}(\rv{S},\rv{Z})]
	\nonumber \\
	&= \sum_{x \in \set{X}} \sum_{t \in \set{T}} \sum_{z \in \set{Z}}  P_{\rv{XT}}(x,t) P_{\rv{Z}|\rv{XT}}(z|x,t) \log \frac{P_{\rv{Z}|\rv{XT}}(z|x,t)}{Q_{\rv{Z}|\rv{X}}(z|x)} \nonumber \\
	& \ \ \ - \mu \sum_{x \in \set{X}}  \sum_{s \in \set{S}}  \sum_{t \in \set{T}} \sum_{z \in \set{Z}}  P_{\rv{XST}}(x,s,t) P_{\rv{Z}|\rv{XT}}(z|x,t) \map{d}(s,z). \label{eq:F_mu_definition}
\end{align}
We obtain that 
\begin{equation}
	L_\mu = \inf\limits_{P_{\rv{Z}|\rv{XT}}, Q_{\rv{Z}|\rv{X}}} F_\mu(P_{\rv{Z}|\rv{XT}}, Q_{\rv{Z}|\rv{X}} ).
\end{equation}

To prove (\ref{eq:minimization_Q}), we only need to note that (\ref{Appendix_BA_proof_3}) is equality if $Q_{\rv{Z}|\rv{X}} = P_{\rv{Z}|\rv{X}}$.

Then, we prove (\ref{eq:minimization_P}). The mutual information $I(\rv{T};\rv{Z}|\rv{X})$ is a convex function of $P_{\rv{Z}|\rv{XT}}$ for fixed $P_{\rv{X}}$ and $P_{\rv{S}}$ by using the same argument in the proof of \cite[Theorem 2.7.4]{cover_elements_nodate}. We consider the Lagrange multiplier to constrain $\sum_{z \in \set{Z} } P_{\rv{Z}|\rv{XT}}(z|x,t)=1 $ and obtain that
\begin{align}
	&\frac{\partial}{\partial P_{\rv{Z}|\rv{XT}}(z|x,t)} \Bigg [ F_\mu(P_{\rv{Z}|\rv{XT}},Q_{\rv{Z}|\rv{X}}) \nonumber \\
	& \ \ \ \ \ \ \ \ \ \ \ + \sum_{x \in \set{X}} \sum_{t \in \set{T}} \lambda(x,t) \sum_{z \in \set{Z}}P_{\rv{Z}|\rv{XT}}(z|x,t) \Bigg ] = 0;
\end{align}
thus
\begin{align}
	P_{\rv{Z}|\rv{XT}} = Q_{\rv{Z}|\rv{X}} \exp \left( \mu \mathbb{E}[\map{d}(\rv{S},z)|x,t] - 1 - \frac{\lambda}{P_{\rv{XT}}(x,t)} \right).
\end{align}
Since the sum of $P_{\rv{Z}|\rv{XT}}$ with respect to $z \in \set{Z}$ is $1$, we conclude the result.

\subsection{Proof of Lemma \ref{thm:converge_to_rate_distortion_function}}

Clearly, $
F_\mu\left( P^{(1)}, Q^{(0)} \right) \ge F_\mu\left( P^{(1)}, Q^{(1)} \right) \ge
F_\mu\left( P^{(2)}, Q^{(1)} \right) \ge \cdots
$. Define the backward probability as follows:
\begin{equation}
	P_{\rv{T}|\rv{XZ}} = \frac{P_{\rv{Z}|\rv{XT}}P_{\rv{T}|\rv{X}}}{Q_{\rv{Z}|\rv{X}}}.
\end{equation}
We have $D\left( P_{\rv{T}|\rv{XZ}}(\cdot|x,z)\big \|P^{(k+1)}_{\rv{T}|\rv{XZ}}(\cdot|x,z)\right) \ge 0$, where $P^{(k+1)}_{\rv{T}|\rv{XZ}}$ is defined as $P_{\rv{T}|\rv{XZ}}$ with $P^{(k+1)}$ and $Q^{(k+1)}$ playing the role of $P_{\rv{Z}|\rv{XT}}$ and $Q_{\rv{Z}|\rv{X}}$, respectively.

For $Q^{(k)}$, arbitrary $P_{\rv{Z}|\rv{XT}}$ and $Q_{\rv{Z}|\rv{X}}$, we notice that 
\begin{align}
	& F_\mu\left(P_{\rv{Z}|\rv{XT}},Q^{(k)} \right) \nonumber \\
	& \ \ \ = F_\mu\left(P_{\rv{Z}|\rv{XT}}, Q_{\rv{Z}|\rv{X}} \right) + D\left( Q_{\rv{Z}|\rv{X}} \big \| Q^{(k)} \big | P_{\rv{X}} \right). \label{Appendix_BA_converge_proof_8}
\end{align}
For $Q^{(k)}$, $P^{(k+1)} = P(Q^{(k)})$, and arbitrary $P_{\rv{Z}|\rv{XT}}$ we consider
\begin{align}
	& F_\mu\left( P_{\rv{Z}|\rv{XT}}, Q^{(k)} \right) \nonumber \\
	= & \sum_{x \in \set{X}} \sum_{t \in \set{T}} \sum_{z \in \set{Z}}  P_{\rv{XT}}(x,t) P_{\rv{Z}|\rv{XT}} \log \frac{P^{(k+1)} (z|x,t)}{Q^{(k)}(z|x)  } \nonumber \\
	&  + D \left( P_{\rv{Z}|\rv{XT}} \big \| P^{(k+1)} \big | P_{\rv{XT}} \right) - \mu \mathbb{E}[\map{d}(\rv{S},\rv{Z})], \label{Appendix_BA_converge_proof_2}
\end{align}
where the expectation is with respect to $P_{\rv{S}} P_{\rv{X}} P_{\rv{T}|\rv{XS}} P_{\rv{Z}|\rv{XT}}$. 
\begin{figure*}[b]
	\normalsize
	\hrulefill
	\begin{align}
		& \sum_{x \in \set{X}} \sum_{t \in \set{T}} \sum_{z \in \set{Z}}  P_{\rv{XT}}(x,t) P_{\rv{Z}|\rv{XT}}(z|x,t) \log \frac{P^{(k+1)} (z|x,t)}{Q^{(k)}(z|x)  } -  \mu \mathbb{E} [\map{d}(\rv{S},\rv{Z})] \nonumber \\
		= & \sum_{x \in \set{X}} \sum_{t \in \set{T}} \sum_{z \in \set{Z}} P_{\rv{XT}}(x,t)  P_{\rv{Z}|\rv{XT}}(z|x,t) \left(  \log \frac{\exp \left( \mu \mathbb{E}[\map{d}(\rv{S},z)|x,t] \right)}{\sum_{a \in \set{Z}}Q^{(k)}(a|x)\exp \left( \mu \mathbb{E}[\map{d}(\rv{S},a)|x,t] \right)} - \mu \sum_{s \in \set{S}} P_{\rv{S}|\rv{XT}}(s|x,t) \map{d}(s,z)\right) \label{Appendix_BA_converge_proof_3} \\
		= & \sum_{x \in \set{X}} \sum_{t \in \set{T}}  P_{\rv{XT}}(x,t)  \log \left( \sum_{a \in \set{Z}}Q^{(k)}(a|x)\exp \left( \mu \mathbb{E}[\map{d}(\rv{S},a)|x,t] \right) \right) \sum_{z \in \set{Z}} P_{\rv{Z}|\rv{XT}}(z|x,t)  \label{Appendix_BA_converge_proof_4} \\
		= & \sum_{x \in \set{X}} \sum_{t \in \set{T}}  P_{\rv{XT}}(x,t)  \log \left( \sum_{a \in \set{Z}}Q^{(k)}(a|x)\exp \left( \mu \mathbb{E}[\map{d}(\rv{S},a)|x,t] \right) \right) \label{Appendix_BA_converge_proof_5}  \\
		= & F_\mu\left( P^{(k+1)}, Q^{(k)} \right) \label{Appendix_BA_converge_proof_6} 
	\end{align}
\end{figure*}
Then, we obtain (\ref{Appendix_BA_converge_proof_3})-(\ref{Appendix_BA_converge_proof_5}) shown at the bottom of this page, where
\begin{enumerate}
	\item[$\bullet$] (\ref{Appendix_BA_converge_proof_3}) invokes (\ref{eq:alternating_P});
	\item[$\bullet$] (\ref{Appendix_BA_converge_proof_4}) follows from expanding the conditional expectation $\mathbb{E}[\map{d}(\rv{S},a)|x,t]$;
	\item[$\bullet$] (\ref{Appendix_BA_converge_proof_5}) holds because $\sum_{z \in \set{Z}} P^{(k)}(z|x,t) =1$.
	\item[$\bullet$] (\ref{Appendix_BA_converge_proof_6}) follows from using the same argument as (\ref{Appendix_BA_converge_proof_3})-(\ref{Appendix_BA_converge_proof_5}) to $F_\mu\left( P^{(k+1)}, Q^{(k)} \right)$. 
\end{enumerate}
Juxtaposing (\ref{Appendix_BA_converge_proof_2}) and (\ref{Appendix_BA_converge_proof_6}) to obtain that
\begin{align}
	&F_\mu\left( P_{\rv{Z}|\rv{XT}}, Q^{(k)} \right) \nonumber \\
	& \ \ \ = F_\mu\left( P^{(k+1)}, Q^{(k)} \right) + D \left( P_{\rv{Z}|\rv{XT}} \big \| P^{(k+1)} \big | P_{\rv{XT}} \right). \label{Appendix_BA_converge_proof_9}
\end{align}
Then, collecting (\ref{Appendix_BA_converge_proof_8}) and (\ref{Appendix_BA_converge_proof_9}), we have the identity
\begin{align}
	& D\left( Q_{\rv{Z}|\rv{X}} \big \| Q^{(k)}\big |P_{\rv{X}}\right) - D\left( Q_{\rv{Z}|\rv{X}} \big \| Q^{(k+1)}\big |P_{\rv{X}}\right) \nonumber \\
	&\ \  = 
	F_\mu \left( P^{(k+1)}, Q^{(k)} \right) - F_\mu \left( P_{\rv{Z}|\rv{XT}}, Q_{\rv{Z}|\rv{X}} \right) \nonumber \\
	&\ \ \ \ \ + \mathbb{E}\left[D\left( P_{\rv{T}|\rv{XZ}}(\cdot|x,z)\big \|P^{(k+1)}_{\rv{T}|\rv{XZ}}(\cdot|x,z)\right)\right]. \label{Appendix_BA_converge_proof_10}
\end{align}
Suppose
\begin{equation}
	F_\mu \left( P_{\rv{Z}|\rv{XT}}, Q_{\rv{Z}|\rv{X}} \right) \le \lim \limits_{n \rightarrow \infty} F_\mu \left( P^{(k+1)}, Q^{(k)} \right),
\end{equation}
where $Q_{\rv{Z}|\rv{X}} = Q(P_{\rv{Z}|\rv{XT}})$.
For any $N > M \ge 1$, (\ref{Appendix_BA_converge_proof_10}) implies
\begin{align}
	0 & \le \sum_{k=M}^{N-1} \left[F_\mu \left( P^{(k+1)}, Q^{(k)} \right) - F_\mu \left( P_{\rv{Z}|\rv{XT}}, Q_{\rv{Z}|\rv{X}} \right) \right] \\
	& \le \sum_{k=M}^{N-1} 	\left[ D\left( Q_{\rv{Z}|\rv{X}} \big \| Q^{(k)}\big |P_{\rv{X}}\right) - D\left( Q_{\rv{Z}|\rv{X}} \big \| Q^{(k+1)}\big |P_{\rv{X}}\right) \right] \\
	& = D\left( Q_{\rv{Z}|\rv{X}} \big \| Q^{(M)}\big |P_{\rv{X}}\right) - D\left( Q_{\rv{Z}|\rv{X}} \big \| Q^{(N)}\big |P_{\rv{X}}\right). \label{Appendix_BA_converge_proof_12}
\end{align}
By inspecting the recursive relations (\ref{eq:alternating_P}) and (\ref{eq:alternating_Q}), we obtain that each term in the RHS of (\ref{Appendix_BA_converge_proof_12}) is bounded. 
Therefore, we find that
\begin{equation}
	\lim \limits_{n \rightarrow \infty } F_\mu \left( P^{(k+1)}, Q^{(k)} \right) = \inf \limits_{P_{\rv{Z}|\rv{XT}}, Q_{\rv{Z}|\rv{X}}} F_\mu \left( P_{\rv{Z}|\rv{XT}}, Q_{\rv{Z}|\rv{X}} \right) = L_\mu. \label{Appendix_BA_converge_proof_13}
\end{equation}

By the Bolzano-Weierstrass Theorem, the sequence $Q^{(k)}$ has a limit $Q^{\star}$ and a subsequence $Q^{(k_i)} \rightarrow Q^{\star}$. Then we have $P^{(k_i+1)} \rightarrow P(Q^{\star}) = P^{\star}$. Substituting this into (\ref{Appendix_BA_converge_proof_13}) and noting that $Q_{\rv{Z}|\rv{X}}$ is obtained by $Q(P_{\rv{Z}|\rv{XT}})$, we have $Q^{\star} = Q(P^{\star})$ and $ F_\mu \left( P^{\star}, Q^{\star} \right) = L_\mu$ as desired.

Additionally, by setting $N = M + 1$ in the RHS of (\ref{Appendix_BA_converge_proof_12}), we obtain that the sequence $D\left( Q^{\star} \big \| Q^{(k)}\big |P_{\rv{X}}\right)$ is monotonic and non-increasing. Furthermore, since $Q^{(k_i)} \rightarrow Q^{\star}$ implies $D\left( Q^{\star} \big \| Q^{(k_i)}\big |P_{\rv{X}}\right) \rightarrow 0$, it follows directly that $D\left( Q^{\star} \big \| Q^{(k)}\big |P_{\rv{X}}\right) \rightarrow 0$. Therefore, $Q^{(k)} \rightarrow Q^{\star}$. This completes the proof.

\printbibliography[heading=bibintoc, title={References}]

%% part of metric entropy.

\begin{IEEEbiographynophoto}{Lugaoze Feng}
	received the B.S. degree from XiDian University, Xian, China, in 2023. He is currently pursuing the Ph.D. degree in communication and information system with Peking University, Beijing. His research interests include information theory and channel coding.
\end{IEEEbiographynophoto}

\begin{IEEEbiographynophoto}{Guocheng Lv}
	received the B.S. degree from Peking University, Beijing, China, in 2006, and the M.S. degree from Peking University, Beijing, China, in 2009. He is currently a Senior Engineer with the School of Electronics, Peking University. His research interests include satellite communication, physical layer modem and non-orthogonal multiple access.
\end{IEEEbiographynophoto}

\begin{IEEEbiographynophoto}{Xunan Li}
	received the B.S. degree in Telecommunications Engineering from Nankai University, Tianjin, China, in 2013, and the Ph.D. degree in Communications and Information System from Peking University, Beijing, China, in 2018. His research interests include communication signal processing and Satellite Communications.
\end{IEEEbiographynophoto}

\begin{IEEEbiographynophoto}{Ye Jin}
	received the B.E. and M.S. degrees from Peking University, Beijing, China, in 1986 and 1989, respectively. He is currently a Professor with the Institute of Modern Communications, Peking University. He has been the Principal Investigator of over 30 funded research projects. His general research interests are in the areas of satellite and wireless communications and networking. Prof. Jin was a recipient of the First Prize of the National Science and Technology Progress Awards of China.
\end{IEEEbiographynophoto}

%\end{CJK}
\newpage

\end{document}